\documentclass[11pt, onecolumn]{IEEEtran}

\usepackage{algorithmic}
\usepackage{amsmath}
\usepackage{amssymb}
\usepackage{amsfonts}
\usepackage{array}
\usepackage{epsfig}
\usepackage{setspace}
\usepackage{cite}
\usepackage{calc}
\usepackage{color}
\usepackage{subfig}
\usepackage{paralist}

\title{Optimal Exact-Regenerating Codes for Distributed Storage at the MSR and MBR Points via a Product-Matrix Construction}

\author{K.~V.~Rashmi, Nihar~B.~Shah, and P.~Vijay~Kumar, {\em Fellow, IEEE}
\thanks{K.~V.~Rashmi, Nihar~B.~Shah and P.~Vijay~Kumar are with the Department of Electrical Communication Engineering, Indian Institute of Science, Bangalore-560012, India (email: \{rashmikv,nihar,vijay\}@ece.iisc.ernet.in). P. Vijay Kumar is also an adjunct faculty member of the Electrical Engineering Systems Department at the University of Southern California, Los Angeles, CA 90089-2565.}
}

\newtheorem{thm}{Theorem}

\newtheorem{lem}[thm]{Lemma}
\newtheorem{cor}[thm]{Corollary}
\newtheorem{remark}{Remark}

\newcommand{\beq}{\begin{equation}}
\newcommand{\eeq}{\end{equation}}
\newcommand{\bea}{\begin{eqnarray}}
\newcommand{\eea}{\end{eqnarray}}
\newcommand{\bean}{\begin{eqnarray*}}
\newcommand{\eean}{\end{eqnarray*}}
\newcommand{\bit}{\begin{itemize}}
\newcommand{\eit}{\end{itemize}}
\newcommand{\ben}{\begin{enumerate}}
\newcommand{\een}{\end{enumerate}}
\newcommand{\blem}{\begin{lem}}
\newcommand{\elem}{\end{lem}}
\newcommand{\bthm}{\begin{thm}}
\newcommand{\ethm}{\end{thm}}
\newcommand{\bpf}{\begin{proof}}
\newcommand{\epf}{\end{proof}}
\begin{document}

\maketitle
\thispagestyle{empty}

\begin{abstract}
Regenerating codes are a class of distributed storage codes that optimally trade the bandwidth needed for repair of a failed node with the amount of data stored per node of the network.  Minimum Storage Regenerating (MSR) codes minimize first, the amount of data stored per node, and then the repair bandwidth, while Minimum Bandwidth Regenerating (MBR) codes carry out the minimization in the reverse order. An $[n,~k,~d]$ regenerating code permits the data to be recovered by connecting to any $k$ of the $n$ nodes in the network, while requiring that repair of a failed node be made possible by connecting (using links of lesser capacity) to any $d$ nodes. Previous, explicit and general constructions of exact-regenerating codes have been confined to the case $n=d+1$.

In this paper, we present optimal, explicit constructions of MBR codes for all feasible values of $[n,~k,~d]$ and MSR codes for all $[n,~k,~d \geq 2k-2]$, using a product-matrix framework. The particular product-matrix nature of the constructions is shown to significantly simplify system operation.  To the best of our knowledge, these are the first constructions of exact-regenerating codes that allow the number $n$ of nodes in the distributed storage network, to be chosen independent of the other parameters.

The paper also contains a simpler description, in the product-matrix framework, of a previously constructed MSR code in which the parameter $d$ satisfies $[n=d+1,~ k,~d \geq 2k-1]$.

\end{abstract}

\section{Introduction}
In a distributed storage system, information pertaining to a data file (the \textit{message}) is dispersed across nodes in a network in such a manner that an end-user can retrieve the data stored by tapping into a subset of the nodes.  A popular option that reduces network congestion and that leads to increased resiliency in the face of  node failures, is to employ erasure coding, for example by calling upon maximum-distance-separable (MDS) codes such as Reed-Solomon (RS) codes.   Let $B$ be the total file size measured in terms of symbols over a finite field $\mathbb{F}_q$ of size $q$.  With RS codes, data is stored across $n$ nodes in the network in such a way that the entire message can be recovered by a data collector by connecting to any $k$ nodes, a process that we will refer to as \textit{data-reconstruction}.  Several distributed storage systems such as RAID-6~\cite{RAID}, OceanStore~\cite{oceanstore} and Total~Recall~\cite{totalRecall} employ such an  erasure-coding option.

\subsection{Regenerating Codes}

Upon failure of an individual node, a self-sustaining data-storage network must necessarily possess the ability to \textit{regenerate} (i.e., repair) a failed node. An obvious means to accomplish this is to permit the replacement node to connect to any $k$ nodes, download the entire message, and extract the data that was stored in the failed node. But downloading the entire $B$ units of data in order to recover the data stored in a single node that stores only a fraction of the entire message is wasteful, and raises the question as to whether there is a better option. Such an option is indeed available and provided by the concept of a \emph{regenerating code} introduced in the pioneering paper by Dimakis et~al.~\cite{DimKan1}.

Conventional RS codes treat each fragment stored in a node as a single symbol belonging to the finite field $\mathbb{F}_q$.  It can be shown that when individual nodes are restricted to perform only linear operations over $\mathbb{F}_q$, the total amount of data download needed to repair a failed node, can be no smaller than $B$, the size of the entire file.
In contrast, regenerating codes are codes over a vector alphabet and hence treat each fragment as being comprised of $\alpha$ symbols over the finite field $\mathbb{F}_q$.  Linear operations over $\mathbb{F}_q$ in this case, permit the transfer of a fraction of the data stored at a particular node.  Apart from this new parameter $\alpha$, two other parameters $d$ and $\beta$ are associated with regenerating codes. Under the definition of regenerating codes introduced in \cite{DimKan1}, a failed node is permitted to connect to an arbitrary set of $d$ of the remaining nodes while downloading $\beta \leq \alpha$ symbols from each node. This process is termed as \textit{regeneration} and the total amount $d\beta$ of data downloaded for repair purposes as the \textit{repair bandwidth}. Further, the set of $d$ nodes aiding in the repair are termed as \textit{helper nodes}. Typically, with a regenerating code, the average repair bandwidth is small compared to the size of the file $B$.

It will be assumed throughout the paper, that whenever mention is made of an $[n,~k,~d]$ regenerating code, the code is such that $k$ and $d$ are the minimum values under which data-reconstruction and regeneration can always be guaranteed. This restricts the range of $d$ to 
\beq k \leq d \leq n-1~. \eeq
The first inequality arises because if the regeneration parameter $d$ were less than the data-reconstruction parameter $k$ then one could, in fact, reconstruct data by connecting to any $d<k$ nodes (treating the data stored in every other node as a function of that stored in these $d$ nodes) thereby contradicting the minimality of $k$. Finally, while a regenerating code over $\mathbb{F}_q$ is associated with the collection of parameters
\[
\{ n, \ k, \ d, \ \alpha, \ \beta, \ B\}~,
\]
it will be found more convenient to regard parameters $\{n,~k,~d\}$ as primary and $\{\alpha,~\beta,~B\}$ as secondary and thus we will make frequent references in the sequel, to a code with these six parameters as an $[n,~k,~d]$ regenerating code having parameter set
$(\alpha,~\beta,~B)$.

\subsection{Cut-Set Bound and the Storage vs Repair-Bandwidth Tradeoff}

A major result in the field of regenerating codes is the proof in \cite{YunDimKan} that uses the cut-set bound of network coding to establish that the parameters of a regenerating code must necessarily satisfy : \bea B & \leq & \sum_{i=0}^{k-1} \min\{\alpha,(d-i)\beta\}~. \label{eq:cut-set bound} \eea 

It is desirable to minimize both $\alpha$ as well as $\beta$ since, minimizing $\alpha$ results in a minimum storage solution, while minimizing $\beta$ (for fixed $d$) results in a storage solution that minimizes repair bandwidth.  As can be deduced from~\eqref{eq:cut-set bound}, it is not possible to minimize both $\alpha$ and $\beta$ simultaneously and thus there is a tradeoff between choices of the parameters $\alpha$ and $\beta$. The two extreme points in this tradeoff are termed the minimum storage regeneration (MSR) and minimum bandwidth regeneration (MBR) points respectively. The parameters $\alpha$ and $\beta$ for the MSR point on the tradeoff can be obtained by first minimizing $\alpha$ and then minimizing $\beta$ to obtain \bea
\alpha & = & \frac{B}{k}~, \nonumber \\
\beta & = & \frac{B}{k(d-k+1)}~. \label{eq:MSR_parameters} \eea
Reversing the order, leads to the MBR point which thus corresponds
to \bea
\beta & = & \frac{2B}{k(2d-k+1)}~, \nonumber \\
\alpha & = & \frac{2dB}{k(2d-k+1)}~.
 \label{eq:MBR_parameters} \eea

We define an \textit{optimal} $[n,~k,~d]$ regenerating code as a code with parameters $(\alpha,~ \beta,~ B)$ satisfying the twin requirements that \begin{enumerate} \item the parameter set $(\alpha,~ \beta, ~B)$ achieves the cut-set bound with equality and \item decreasing either $\alpha$ or $\beta$ will result in a new parameter set that violates the cut set bound. \end{enumerate} An MSR code is then defined as an $[n,~k,~d]$ regenerating code whose parameters $(\alpha, ~\beta,~ B)$ satisfy \eqref{eq:MSR_parameters} and similarly, an MBR code as one with parameters $(\alpha, ~\beta,~ B)$ satisfying \eqref{eq:MBR_parameters}.  Clearly, both MSR and MBR codes are optimal regenerating codes.

\subsection{Striping of Data}
The nature of the cut-set bound permits a divide-and-conquer approach to be used in the application of optimal regenerating codes to large file sizes, thereby simplifying system implementation.  This is explained below.

Given an optimal $[n,~k,~d]$ regenerating code with parameter set $(\alpha,\beta,B)$, a second optimal regenerating code with parameter set $(\alpha^{'}=\delta \alpha , \ \beta^{'}=\delta \beta , \ B^{'}=\delta B)$ for any positive integer $\delta$ can be constructed, by dividing the $\delta B$ message symbols into $\delta$ groups of $B$ symbols each, and applying the $(\alpha,~\beta,~B)$ code to each group independently. Secondly, a common feature of both MSR and MBR regenerating codes is that in either case, their parameter set $(\alpha,\beta,B)$ is such that both $\alpha$ and $B$ are multiples of $\beta$ and further that $\frac{\alpha}{\beta}$, $\frac{B}{\beta}$ are functions only of $n$, $k$ and $d$. It follows that if one can construct an (optimal) $[n,~k,~d]$ MSR or MBR code with $\beta=1$, then one can construct an (optimal) $[n,~k,~d]$ MSR or MBR code for any larger value of $\beta$. In addition, from a practical standpoint, a code constructed through concatenation of codes for a smaller $\beta$ will in general, be of lesser complexity (see Section~\ref{subsec:complexity}). For these reasons, in the present paper we design codes for the case of $\beta=1$. Thus, throughout the remainder of the paper, we will assume that $\beta=1$. In the terminology of distributed storage, such a process is called striping.

We document below the values of $\alpha$ and $B$ of MSR and MBR codes respectively,  when $\beta=1$: \bea \alpha&=&d-k+1~, \label{eq:MSR_alpha_beta1} \\ B &=& k (d-k+1) \label{eq:MSR_B_beta1} \eea for MSR codes and \bea \alpha &=& d~,\label{eq:MBR_alpha_beta1}\\ B &=& kd - {k\choose 2}~ \label{eq:MBR_B_beta1} \eea in the case of
MBR codes.

\subsection{Additional Terminology}

\paragraph{Exact versus Functional Regeneration} In the context of a regenerating code, by functional regeneration of a failed node $f$, we will mean, replacement of the failed node by a new node $f'$ in such a way that following replacement, the resulting network of $n$ nodes continues to possess the data-reconstruction and regeneration properties.   In contrast, by exact-regeneration, we mean replacement of a failed node $f$ by a replacement node $f'$ that stores exactly the same data as was stored in node $f$.  We will use the term {\em exact-regenerating code} to denote a regenerating code that has the capability of exactly regenerating each instance of a failed node. Clearly where it is possible, an exact-regeneration code is to be preferred over a functional-regeneration code since, under functional regeneration, there is need for the network to inform all nodes in the network of the replacement, whereas this is clearly unnecessary under exact-regeneration.

\paragraph{Systematic Regenerating Codes} A systematic regenerating code can be defined as a regenerating code designed in such a way that the $B$ message symbols are explicitly present amongst the $k \alpha$ code symbols stored in a select set of $k$ nodes, termed as the systematic nodes. Clearly, in the case of systematic regenerating codes, exact-regeneration of (the systematic portion of the data stored in) the systematic nodes is mandated.

\paragraph{Linear Regenerating Codes}
A linear regenerating code is defined as a regenerating code in which \ben \item the code symbols stored in each node are linear combinations over $\mathbb{F}_q$ of the $B$ message symbols $\{u_i\}_{i=1}^{B}$, \item the symbols passed by a helper node $h$ to aid in the regeneration of a failed node $f$ are linear over $\mathbb{F}_q$ in the $\alpha$ symbols stored in node $h$. \een

It follows as an easy consequence, that linear operations suffice for a data collector to recover the data from the $k \alpha$ code symbols stored in the $k$ nodes that it has connected to. Similarly, the replacement node for a failed node $f$, performs linear operations on the $d$ symbols passed on to it by the $d$ helper nodes $\{h_i\}_{i=1}^d$ aiding in the regeneration.

\subsection{Results of the Present Paper}

While prior work is described in greater detail in Section~\ref{sec:priorWork}, we begin by providing a context for the results presented here.

\paragraph*{Background} To-date, explicit and general constructions for exact-regenerating codes at the MSR point have been found only for the case $[n=d+1,~k,~d \geq 2k-1]$. Similarly at the MBR point, the only explicit code to previously have been constructed is for the case $d=n-1$.  Thus, all existing code constructions limit the total number of nodes $n$ in the system to $d+1$. This is restrictive since in this case, the system can handle only a single node failure at a time. Also, such a system does not permit additional storage nodes to be brought into the system.

A second open problem in this area that has recently drawn attention is as to whether or not the storage-repair bandwidth tradeoff is achievable under the additional requirement of exact-regeneration. It has previously been shown that the MSR point is not achievable for any $[n,~k,~d\leq 2k-3]$ with $\beta=1$, but is achievable for all parameters $[n,~k,~d]$ when $B$ (and hence $\beta$ as well) is allowed to approach
infinity. %It remains an open problem whether the MSR and MBR points are achievable for a general set of parameters. By presenting explicit constructions, the present paper answers this question in the affirmative. In the constructions presented here, the file size is small and as dictated by the cut-set bound.

\paragraph*{Results Presented in Present Paper} In this paper, (optimal) explicit constructions of exact-regenerating MBR codes for all feasible values of $[n,~k,~d]$ and exact-regenerating MSR codes for all $[n,~k,~d \geq 2k-2]$ are presented.  The constructions are of a product-matrix nature that is shown to significantly simplify operation of the distributed storage network. The constructions presented prove that the MBR point for exact-regeneration can be achieved for all values of the parameters and that the MSR point can be achieved for all parameters satisfying $d \geq 2k-2$. In both constructions, the message size is as dictated by cut-set bound. The paper also contains a simpler description, in the product-matrix framework, of an MSR code for the parameters $[n=d+1,~k,~d\geq 2k-1]$ that was previously constructed in~\cite{ourITW,Changho}.

A brief overview of prior work in this field is provided in the next section, Section~\ref{sec:priorWork}. The product-matrix framework underlying the code construction is described in Section~\ref{sec:prod_mx}. An exact-regenerating MBR code for all feasible values of the parameters $[n, \ k, \ d]$ is presented in Section~\ref{sec:MBR_prodmx}, and an exact-regenerating MSR code for all $[n,~k,~d \geq 2k-2]$ is presented in Section~\ref{sec:MSR_prodmx}. Implementation advantages of the particular product-matrix nature of the code constructions provided here are described in Section~\ref{sec:sys_adv}. The final section, Section~\ref{sec:conclusion}, draws conclusions. Appendix~\ref{app:MISER_prodmx} contains a simpler description, in the product-matrix framework, of an MSR code with parameter $d$ satisfying $[n=d+1,~k,~d \geq  2k-1]$, that was previously constructed in \cite{ourITW,Changho}.

\section{Prior Work}\label{sec:priorWork}
The concept of regenerating codes was introduced in~\cite{DimKan1}, where it was shown that permitting the storage nodes to store more than $B/k$ units of data helps in reducing the repair bandwidth. Several distributed systems were analyzed, and estimates of the mean node availability in such systems obtained. Using these values, it was shown through simulation, that regenerating codes can reduce repair bandwidth compared to other designs, while simplifying system
architecture.

The problem of minimizing repair bandwidth for functional repair of a failed storage node is considered in~\cite{DimKan1,YunDimKan}. Here, the evolution of the storage network through a sequence of failures and regenerations is represented as a network, with all possible data collectors represented as sinks. The data-reconstruction requirement is formulated as a multicast network coding problem, with the network having an infinite number of nodes. The cut-set analysis of this network leads to the relation between the parameters of a regenerating code given in equation~\eqref{eq:cut-set bound}. It can be seen that there is a tradeoff between the choice of the parameters $\alpha$ and $\beta$ for a fixed $B$ and this is termed as the storage-repair bandwidth tradeoff. It has been shown~(\!\!\cite{YunDimKan,WuAchievable}) that this tradeoff is achievable under functional-regeneration. However, the coding schemes suggested are not explicit and require large field size. The journal version~\cite{YunDimKanJournal} also contains a handcrafted functional-regenerating code for the MSR point with  $[n=4, \ k=2,\ d=3]$.

A study of the computational complexity of regenerating codes is carried out in~\cite{Complexi}, in the context of random linear regenerating codes that achieve functional repair.

The problem of exact-regeneration was first considered independently in~\cite{WuDimISIT,ourAllerton} and~\cite{DimSearch}. In~\cite{WuDimISIT}, it is shown that the MSR point is achievable under exact-regeneration when $(k=2, \ d=n-1)$. The coding scheme proposed is based on the concept of interference alignment developed in the context of wireless communication. However, the construction is not explicit and has a large field size requirement. In~\cite{DimSearch}, the authors carry out a computer search to find exact-regenerating codes at the MSR point, resulting in identification of codes with parameters $[n=5, \ k=3, \ d=4]$.

The first, explicit construction of regenerating codes for a general set of parameters was provided for the MBR point in~\cite{ourAllerton} with $d=n-1$ and arbitrary $k$.  These codes have low regeneration complexity as no computation is involved during the exact-regeneration of a failed node. The field size required is of the order of $n^2$. In addition~\cite{ourAllerton} (see also \cite{ourMISERjournal}) also contains the construction of an explicit MSR code for $d=k+1$, that performs approximately-exact-regeneration of all failed nodes, i.e., regeneration where a part of the code is exactly regenerated, and the remaining is functionally regenerated (it is shown subsequently in~\cite{ourITW,ourMISERjournal} that exact-regeneration is not possible, when $k>4$, for the set of parameters considered therein).

MSR codes performing a hybrid of exact and functional-regeneration are provided in~\cite{WuArxiv}, for the parameters $d=k+1$ and $n>2k$. The codes given even here are non-explicit, and have high complexity and large field-size requirement.

A code structure that guarantees exact-regeneration of just the systematic nodes is provided in~\cite{ourITW}, for the MSR point with parameters $[n=d+1,~k,~d\geq 2k-1]$. This code makes use of interference alignment, and is termed as the `MISER' code in journal-submission version~\cite{ourMISERjournal} of~\cite{ourITW}. Subsequently, it was shown in~\cite{Changho} that for this set of parameters, the code introduced in~\cite{ourITW} for exact-regeneration of only the systematic nodes can also be used to repair the non-systematic (parity) node failures exactly provided repair construction schemes are appropriately designed. Such an explicit repair scheme is indeed designed and presented in~\cite{Changho}. The paper~\cite{Changho} also contains an exact-regenerating MSR code for parameter set $[n=5, \ k=3, \ d=4]$.

A proof of non-achievability of the cut-set bound on exact-regeneration at the MSR point, for the parameters $[n,~k,~d<2k-3]$ when $\beta=1$, is provided in~\cite{ourMISERjournal}. On the other hand, the MSR point is shown to be achievable in the limiting case of $B$ approaching infinity (i.e., $\beta$ approaching infinity)
in~\cite{Jafar_arxiv,Changho_arxiv_intfalign}.

A flexible setup for regenerating codes is described in~\cite{ourFlex}, where a data collector (or a replacement node) can perform data-reconstruction (or regeneration) irrespective of the number of nodes to which it connects, provided the total data downloaded exceeds a certain threshold.

In~\cite{ourAllertonJournal}, the authors establish that essentially all points on the interior of the tradeoff (i.e., points other than MSR and MBR) are not achievable under exact-regeneration. 

\section{The Common Product-Matrix Framework}\label{sec:prod_mx}
The constructions described in this paper, follow a common product-matrix framework. Under this framework, each codeword in the distributed storage code can be represented by an $(n \times \alpha)$ \textit{code matrix} $C$ whose $i^{th}$ row $\underline{c}_i^t$ contains the $\alpha$ symbols stored by the $i^{th}$ node. Each code matrix is the product \bea C = \Psi M \eea
of an $(n \times d)$ \textit{encoding matrix} $\Psi$ and an $(d \times \alpha)$ \textit{message matrix} $M$. The entries of the matrix $\Psi$ are fixed a priori and are independent of the message symbols. The message matrix $M$ contains the $B$ message symbols, with some symbols possibly repeated. We will refer to the $i^{th}$ row $\underline{\psi}_i^t$ of $\Psi$ as the encoding vector of node $i$ as it is this vector that is used to encode the message into the form in which it is stored within the $i^{th}$ node:
\beq \underline{c}_i^t = \underline{\psi}_i^t M~, \eeq
where the superscript `$t$' is used to denote the transpose of a matrix. Throughout this paper, we consider all symbols to belong to a finite field $\mathbb{F}_q$ of size $q$.

This common structure of the code matrices leads to common architectures for both data-reconstruction and exact-regeneration, as explained in greater detail below. It also endows the codes with implementation advantages that are discussed in Section~\ref{sec:sys_adv}.

%Observe that the product-matrix framework generates the $\alpha$ symbols stored in a node via the product of a fixed encoding vector with a message matrix. This is in contrast to a general linear code that would require the product of $\alpha$ encoding vectors (from the generator matrix) with a message vector.

Data-reconstruction amounts to recovering the message matrix $M$ from the $k\alpha$ symbols obtained from an arbitrary set of $k$ storage nodes. Let us denote the set of $k$ nodes to which the data collector connects as $\{i_1,\ldots,i_k\}$. The $j^{th}$ node in this set passes on the message vector $\underline{\psi}_{i_j}^t M$ to
the data collector. The data collector thus obtains the product matrix \[ \Psi_{\text{\tiny DC}} M~,
\] where $\Psi_{\text{\tiny DC}}$ is the submatrix of $\Psi$ consisting of the $k$ rows $\{i_1,\ldots,i_k\}$.
It then uses the properties of the matrices $\Psi$ and $M$ to recover the message. The precise procedure for recovering $M$ is a function of  the particular construction.

As noted above, each node in the network is associated to a distinct $(d \times 1)$ encoding vector $\underline{\psi}_i$. In the regeneration process, we will need to call upon a related vector $\underline{\mu}_i$ of length $\alpha$, that contains a subset of the components of $\underline{\psi}_i$. To regenerate a failed node $f$, the node replacing the failed node connects to an arbitrary subset $\{h_1,\ldots,h_d\}$ of $d$ storage nodes which we will refer to as the $d$ {\em helper nodes}.  Each helper node passes on the inner product of the $\alpha$ symbols stored in it with $\underline{\mu}_f$, to the replacement node: the helper node $h_j$ passes \[ \underline{\psi}_{h_j}^t M\underline{\mu}_f~.\] The replacement node thus obtains the product matrix \[ \Psi_{\text{repair}} M \underline{\mu}_f ~,\] where $\Psi_{\text{repair}}$ is the submatrix of $\Psi$ consisting of the $d$ rows $\{h_1,\ldots,h_d\}$. From this it turns out, as will be shown subsequently, that one can recover the desired symbols. Here again, the precise procedure is dependent on the particular construction.

\begin{remark}
An important feature of the product-matrix construction presented here, is that each of the nodes $h_j$ participating in the regeneration of node $f$, needs only have knowledge of the encoding vector of the failed node $f$ and not the identity of the other nodes participating in the regeneration. This significantly simplifies the operation of the system.
\end{remark}

\textit{Systematic Codes:} The following theorem shows that any linear exact-regenerating code can be converted to a systematic form via a linear remapping of the symbols. The proof of the theorem may be found in Appendix~\ref{app:gen_mx}.

\begin{thm}
Any linear exact-regenerating code can be converted to a systematic form via a linear remapping of the message symbols. Furthermore, the resulting code is also linear and possesses the data-reconstruction and exact-regeneration properties of the original code.
\end{thm}

Thus, all codes provided in the present paper can be converted to a systematic form via a linear remapping of the message symbols. Specific details on the product-matrix MBR and MSR codes in systematic form are provided in the respective sections, Sections~\ref{sec:MBR_prodmx} and~\ref{sec:MSR_prodmx}.

\section{The Product-Matrix MBR Code Construction} \label{sec:MBR_prodmx}

In this section, we identify the specific make-up of the encoding matrix $\Psi$ and the message matrix $M$ that results in an $[n,k,d]$ MBR code with $\beta=1$.   A notable feature of the construction is that it is applicable to all feasible values of $[n,k,d]$, i.e., all $n,k,d$ satisfying $k \leq d \leq n-1$.  Since the code is required to be an MBR code with $\beta=1$, it must possess the data-reconstruction and exact-regeneration properties required of a regenerating code, and in addition, have parameters $\{\alpha, B\}$ that satisfy equations \eqref{eq:MBR_alpha_beta1} and \eqref{eq:MBR_B_beta1}.  Equation \eqref{eq:MBR_B_beta1} can be rewritten in the form:
\bean B & = & {k+1 \choose 2} + k(d-k)~. \eean  Thus the parameter set of the desired $[n,k,d]$ MBR code is $(\alpha=d, \beta=1, B = {k+1 \choose 2} + k(d-k))$.

Let $S$ be a $(k \times k)$ matrix constructed so that the ${k+1 \choose 2}$ entries in the upper-triangular half of the matrix are filled up by ${k+1 \choose 2}$ distinct message symbols drawn from the set $\{u_i\}_{i=1}^B$.  The ${k \choose 2}$ entries in the strictly lower-triangular portion of the matrix are then chosen so as to make the matrix $S$ a symmetric matrix.   The remaining $k(d-k)$ message symbols are used to fill up a second $(k \times (d-k))$ matrix $T$.  The message matrix $M$ is then defined as the  $(d \times d)$ symmetric matrix given by
\beq M = \left[
\begin{tabular}{>{$}c<{$} >{$}c<{$}}
S & T\\
T^t  & 0
\end{tabular}
\right]. \eeq  The symmetry of the matrix will be found to be instrumental when enabling node repair. Next, define the encoding matrix
$\Psi$ to be any $(n \times d)$ matrix of the form
\[ \Psi = \left[\begin{tabular}{>{$}c<{$}>{$}c<{$}}
\Phi & \Delta
\end{tabular}
\right], \] where $\Phi$ and $\Delta$ are $(n \times k)$ and $(n \times
(d-k))$ matrices respectively, chosen in such a way that
\begin{enumerate}
\item any $d$ rows of $\Psi$ are linearly independent,
\item any $k$ rows of $\Phi$ are linearly independent.
\end{enumerate}
The above requirements can be met, for example, by choosing $\Psi$
to be either a Cauchy~\cite{cauchy} or else a Vandermonde matrix.\footnote{Over a large finite field, every matrix $\Psi$ in general position will also satisfy these requirements.} 
As per the product-matrix framework, the code matrix is then given by $C \ = \  \Psi M $.

The two theorems below establish that the code presented is an $[n,k,d]$ MBR code by establishing respectively, the exact-regeneration and data-reconstruction  properties of the code.

\begin{thm}[MBR Exact-Regeneration] \label{thm:MBR_regen}
In the code presented, exact-regeneration of any failed node can be achieved by connecting to any $d$ of the $(n-1)$ remaining nodes.
\end{thm}

\begin{IEEEproof} Let $\underline{\psi}_f^t$ be the row of $\Psi$ corresponding to the failed node $f$. Thus the $d$ symbols stored in the failed node correspond to the vector \beq \underline{\psi}_f^t M  . \eeq The replacement for the failed node $f$ connects to an arbitrary set $\{h_j \mid j=1,\ldots,d\}$ of $d$ helper nodes.  Upon being contacted by the replacement node, the helper node $h_j$ computes the inner product \[\underline{\psi}_{h_j}^t M\underline{\psi}_f\] and passes on this value to the replacement node.  Thus, in the present construction, the vector $\underline{\mu}_f$ equals $\underline{\psi}_f$ itself. The replacement node thus obtains the $d$ symbols $\Psi_{\text{repair}}M \underline{\psi}_f$ from the $d$ helper nodes, where
\bean
\Psi_{\text{repair}} & = & \left[
\begin{array}{c} \underline{\psi}_{h_1}^t \\
\underline{\psi}_{h_2}^t \\
 \vdots \\
 \underline{\psi}_{h_d}^t \end{array} \right] .
 \eean
By construction, the $(d \times d)$ matrix $\Psi_{\text{repair}}$ is invertible.  Thus, the replacement node recovers $M \underline{\psi}_f$ through multiplication on the left by $\Psi_{\text{repair}}^{-1}$.  Since $M$ is symmetric, \beq \left( M \underline{\psi}_f \right)^t \ = \
\underline{\psi}_f^t M~, \eeq and this is precisely the data previously stored in the failed node.
\end{IEEEproof}

\begin{thm}[MBR Data-Reconstruction] \label{thm:MBR_reconstruction} In the code presented, all the $B$ message symbols can be recovered by connecting to any $k$ nodes, i.e., the message symbols can be recovered through linear operations on the entries of any $k$ rows of the matrix $C$.
\end{thm}

\begin{IEEEproof} Let \beq \Psi_{\text{\tiny DC}} = \left[
\begin{tabular}{>{$}c<{$}>{$}c<{$}} \Phi_{\text{\tiny DC}} &
\Delta_{\text{\tiny DC}}
\end{tabular}
\right] \eeq be the $(k \times d)$ submatrix of $\Psi$, corresponding to the $k$ rows of $\Psi$ to which the data collector connects.  Thus the data collector has access to the symbols \bea
\Psi_{\text{\tiny DC}} M = \left[
\begin{tabular}{>{$}c<{$}>{$}c<{$}} \Phi_{\text{\tiny DC}}
S  + \Delta_{\text{\tiny DC}} T^t  &
\Phi_{\text{\tiny DC}} T
\end{tabular}
\right].
\eea
By construction, $\Phi_{\text{\tiny DC}}$ is a non-singular matrix.  Hence, by multiplying the matrix $\Psi_{\text{\tiny DC}} M$ on the left by $\Phi_{\text{\tiny DC}}^{-1}$, one can recover first $T$ and subsequently, $S$.
\end{IEEEproof}

\subsection{An Example for the Product-Matrix MBR Code} \label{sec:MBR_example}

Let $n=6, \ k=3, \ d=4$.  Then $\alpha=d=4$ and $B=9$.  Let us choose $q=7$ so we are operating over $\mathbb{F}_7$.  The matrices $S$ and $T$ are filled up by the $9$ message symbols $\{u_i\}_{i=1}^9$ as shown below:
\beq S \ = \ \begin{bmatrix}
u_1 & u_2 & u_3 \\
u_2 & u_4 & u_5 \\
u_3 & u_5 & u_6 \end{bmatrix}, \quad T =\begin{bmatrix}
u_7 \\
u_8 \\
u_9
\end{bmatrix},
\eeq
so that the message matrix $M$ is given by
\beq M \ = \ \begin{bmatrix}
                   u_1 & u_2 & u_3 & u_7 \\
           u_2 & u_4 & u_5 & u_8 \\
           u_3 & u_5 & u_6 & u_9 \\
           u_7 & u_8 & u_9 & 0
           \end{bmatrix} .
\eeq
We choose $\Psi$ to be the $(6 \times 4)$ Vandermonde matrix over $\mathbb{F}_7$ given by \beq \Psi= \begin{bmatrix}
            1 & 1 & 1 & 1\\
            1 & 2 & 4 & 1\\
            1 & 3 & 2 & 6\\
            1 & 4 & 2 & 1\\
            1 & 5 & 4 & 6\\
            1 & 6 & 1 & 6
           \end{bmatrix}.  \eeq
           Fig.~\ref{fig:MBR_example} shows at the top, the $(6 \times 4)$ code matrix $C=\Psi M$ with entries expressed as functions of the message symbols $\{u_i \}_{i=1}^9$.  The rest of the figure explains how exact-regeneration of failed node $1$ takes place. To regenerate node $1$, the helper nodes (nodes $2,4,5,6$ in the example), pass on their respective inner products $\underline{\psi}_{\ell}^t M [1 \ 1 \ 1 \ 1]^t$ for $\ell=2,4,5,6$. The replacement node then recovers the data stored in the failed node by multiplying by $\Psi_{\text{repair}}^{-1}$ where
           \beq
           \Psi_{\text{repair}} \ = \
           \begin{bmatrix}
            1 & 2 & 4 & 1\\
            1 & 4 & 2 & 1\\
            1 & 5 & 4 & 6\\
            1 & 6 & 1 & 6
           \end{bmatrix}  \eeq
as explained in the proof of Theorem~\ref{thm:MBR_regen} above.

\begin{figure*}
\centering
\includegraphics[trim=.6in 5.6in 3.4in 0.6in, clip, width=.95\textwidth]{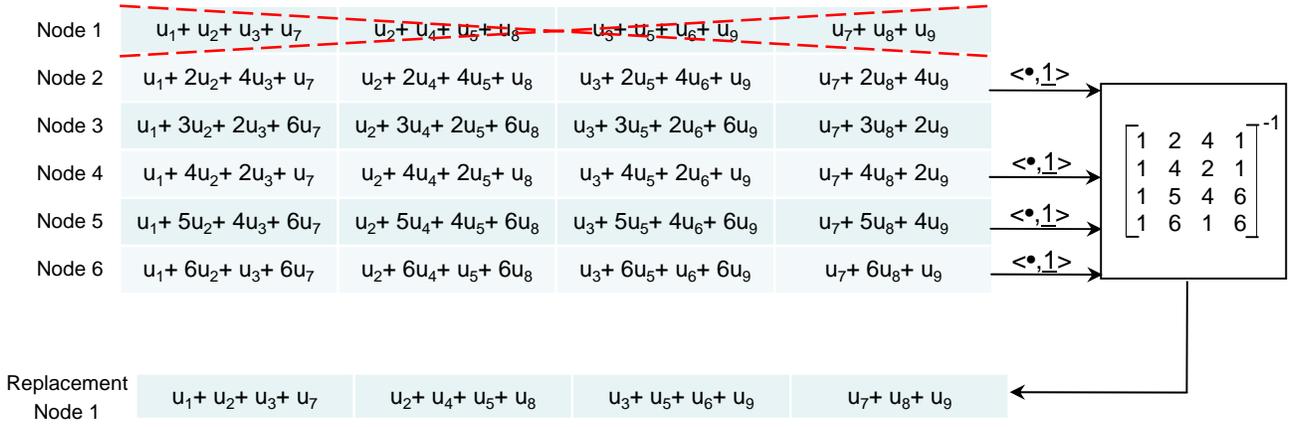}
\caption{An example for the MBR code construction: On failure of node $1$, the replacement node downloads one symbol each from nodes $2, \ 4, \ 5$ and $6$, using which node $1$ is exactly regenerated. The notation $<.,\underline{1}>$ indicates an inner product of the stored symbols with the vector $[1~1~1~1]^t$.} \label{fig:MBR_example}
\end{figure*}

\subsection{Systematic Version of the Code} As pointed out in Section~\ref{sec:prod_mx}, any exact-regenerating code can be made systematic through a non-singular transformation of the message symbols.  In the present case, there is a simpler approach, in which the matrix $\Psi$ can be chosen in such a way that the code is automatically systematic. We simply make the choice: \beq \Psi = \left[
\begin{tabular}{>{$}c<{$} >{$}c<{$}}
I_k &  0 \\
\tilde{\Phi} & \tilde{\Delta}
\end{tabular}
\right], \eeq where $I_k$ is the $(k \times k)$ identity matrix, $0$ is a $(k \times (d-k)$ zero matrix, $\tilde{\Phi}$ and $\tilde{\Delta}$ are matrices of sizes $((n-k) \times k)$ and $((n-k) \times (d-k))$ respectively, such that $\left[\tilde{\Phi}~~~\tilde{\Delta}\right]$ is a Cauchy matrix\footnote{In general, any matrix, all of whose submatrices are of full rank, will suffice.}. Clearly the code is systematic.  It can be verified that the matrix $\Psi $ has the properties listed just above Theorem~\ref{thm:MBR_regen}.

\section{The Product-Matrix MSR Code Construction}\label{sec:MSR_prodmx}

In this section, we identify the specific make-up of the encoding matrix $\Psi$ and the message matrix $M$ that results in an $[n,~k,~d]$ MSR code with $\beta=1$.   The construction applies to all parameters $[n,~k,~d\geq 2k-2]$.\footnote{As mentioned previously, it is impossible to construct linear MSR codes for the case of $d < 2k-3$ when $\beta=1$.} Since the code is required to be an MSR code with $\beta=1$, it must possess the data-reconstruction and exact-regeneration properties required of a regenerating code, and in addition, have parameters $\{\alpha, B\}$ that satisfy equations \eqref{eq:MSR_alpha_beta1} and \eqref{eq:MSR_B_beta1}.
We begin by constructing an MSR code in the product-matrix format for $d=2k-2$ and will show in Section~\ref{sec:MSR_extension} how this can be very naturally extended to yield codes with  $d > 2k-2$.

At the MSR point with $d=2k-2$ we have
\bea \alpha &=& d-k+1 \ = \  k-1~, \label{eq:alpha_k_minus_1} \eea
and hence \beq d= 2\alpha~. \label{eq:d_2alpha}\eeq
Also, \bea B &=& k \alpha \ = \  \alpha (\alpha+1)~. \label{eq:MSR_B_2k_2}\eea

We define the $(d \times \alpha)$ message matrix $M$ as \beq
M = \left[
\begin{tabular}{>{$}c<{$}}
S_1 \\
S_2  \\
\end{tabular}
\right], \eeq
where $S_1$ and $S_2$ are $(\alpha \times \alpha)$ symmetric matrices constructed such that the ${\alpha+1 \choose 2}$ entries in the upper-triangular part of each of the two matrices are filled up by ${\alpha+1 \choose 2}$ distinct message symbols. Thus, all the $B=\alpha(\alpha+1)$ message symbols are contained in the two matrices $S_1$ and $S_2$. The entries in the strictly lower-triangular portion of the two matrices $S_1$ and $S_2$ are chosen so as to make the matrices $S_1$ and $S_2$ symmetric.
% $\frac{\alpha(\alpha+1)}{2}$

Next, we define the encoding matrix $\Psi$ to be the  $(n \times d)$ matrix given by
\bea
\Psi = \begin{bmatrix}
    \Phi & \Lambda\Phi
    \end{bmatrix},
\eea
where $\Phi$ is an $(n \times \alpha)$ matrix and $\Lambda$ is an $(n \times n)$ diagonal matrix.  The elements of $\Psi$ are chosen such that the following conditions are satisfied:
\begin{enumerate}
\item any $d$ rows of $\Psi$ are linearly independent,
\item any $\alpha$ rows of $\Phi$ are linearly independent and
\item the $n$ diagonal elements of $\Lambda$ are distinct.
\end{enumerate}

The above requirements can be met, for example, by choosing $\Psi$ to be a Vandermonde matrix with elements chosen carefully to satisfy the third condition. Then under our code-construction framework, the $i^{th}$ row of the $(n \times \alpha)$ product matrix $C \ = \  \Psi M $, contains the $\alpha$ code symbols stored by the $i^{th}$ node.

The two theorems below establish that the code presented is an $[n,k,d]$ MSR code by establishing respectively, the exact-regeneration and data-reconstruction  properties of the code.

\begin{thm}[MSR Exact-Regeneration] In the code presented, exact-regeneration of any failed node can be achieved by connecting to any $d=2k-2$ of the remaining $(n-1)$ nodes.\label{thm:MSR_regen}
\end{thm}

\begin{IEEEproof} Let $\left[ \underline{\phi}_f^t \ \ \ \lambda_f \underline{\phi}_f^t\right]$ be the row of $\Psi$ corresponding to the failed node. Thus the $\alpha$ symbols stored in the failed node
 were
\beq \left[\underline{\phi}_f^t \ \ \ \lambda_f \underline{\phi}_f^t\right] M \ = \
\underline{\phi}_f^t S_{1}+ \lambda_f \underline{\phi}_f^t S_{2}~. \eeq

The replacement for the failed node $f$ connects to an arbitrary set $\{h_j \mid j=1,\ldots,d\}$ of $d$ helper nodes. Upon being contacted by the replacement node, the helper node $h_j$ computes the inner product $\underline{\psi}_{h_j}^t M \underline{\phi}_f$ and passes on this value to the replacement node. Thus, in the present construction, the vector $\underline{\mu}_f$ equals $\underline{\phi}_f$.  The replacement node thus obtains the $d$ symbols $\Psi_{\text{repair}}M \underline{\phi}_f$ from the $d$ helper nodes, where
\bean
\Psi_{\text{repair}} & = & \left[
\begin{array}{c} \underline{\psi}_{h_1}^t \\
\underline{\psi}_{h_2}^t \\
 \vdots \\
 \underline{\psi}_{h_d}^t \end{array} \right] .
 \eean
By construction, the $(d \times d)$ matrix $\Psi_{\text{repair}}$
is invertible.   Thus the replacement node now has access to \bea M \underline{\phi}_f &
= & \begin{bmatrix}
    S_{1} \underline{\phi}_f \nonumber \\
    S_{2}\underline{\phi}_f
    \end{bmatrix} . \eea

As $S_{1}$ and $S_{2}$ are symmetric matrices, the replacement node has thus acquired through transposition, both $\underline{\phi}_f^t S_{1}$ and $\underline{\phi}_f^t S_{2}$. Using this it can obtain
\beq \underline{\phi}_f^t S_{1}+ \lambda_f \underline{\phi}_f^t S_{2}~, \eeq
which is precisely the data previously stored in the failed node.
\end{IEEEproof}

\begin{thm}[MSR Data-Reconstruction] In the code presented, all the $B$ message symbols can be recovered by connecting to any $k$ nodes, i.e., the message symbols can be recovered through linear operations on the entries of any $k$ rows of the code matrix $C$.\label{thm:MSR_recon}
\end{thm}
\begin{IEEEproof}
Let \beq \Psi_{\text{\tiny DC}} =\begin{bmatrix}
    \Phi_{\text{\tiny DC}} & \Lambda_{\text{\tiny DC}}\Phi_{\text{\tiny DC}}
    \end{bmatrix}
\eeq be the $(k \times d)$ submatrix of $\Psi$, containing the $k$
rows of $\Psi$ which correspond to the $k$ nodes to which the data collector
connects.  Hence the data collector obtains the symbols  \bea \Psi_{\text{\tiny
DC}} M &=& \begin{bmatrix}
    \Phi_{\text{\tiny DC}} & \Lambda_{\text{\tiny DC}}\Phi_{\text{\tiny DC}}
    \end{bmatrix}
    \begin{bmatrix}
    S_{1} \nonumber \\
    S_{2}
    \end{bmatrix} \\
&=& \begin{bmatrix}
    \Phi_{\text{\tiny DC}} S_{1} + \Lambda_{\text{\tiny DC}}\Phi_{\text{\tiny DC}}S_{2}
   \end{bmatrix}. \eea
The data collector can post-multiply this term with $\Phi_{\text{\tiny DC}}^T$ to obtain
\bea
    &\begin{bmatrix}
    \Phi_{\text{\tiny DC}} S_{1} + \Lambda_{\text{\tiny DC}}\Phi_{\text{\tiny DC}} S_{2}
   \end{bmatrix}
    \begin{matrix}
    \Phi_{\text{\tiny DC}}^{T}
    \end{matrix} \nonumber \\
 &= \Phi_{\text{\tiny DC}} S_{1} \Phi_{\text{\tiny DC}}^{T}+ \Lambda_{\text{\tiny DC}}\Phi_{\text{\tiny DC}} S_{2} \Phi_{\text{\tiny DC}}^{T}~.
\eea
Next, let the matrices $P$ and $Q$ be defined as
\bea
P = \Phi_{\text{\tiny DC}} S_{1} \Phi_{\text{\tiny DC}}^{T}~, \\
Q = \Phi_{\text{\tiny DC}} S_{2} \Phi_{\text{\tiny DC}}^{T}~.
\eea
As $S_{1}$ and $S_{2}$ are symmetric, the same is true of the matrices $P$ and $Q$.    In terms of $P$ and $Q$, the data collector has access to the symbols of the matrix
\beq
P + \Lambda_{\text{\tiny DC}} Q~.\eeq
The $(i,\ j)^{\text{th}}$, $1 \leq i,j \leq k$, element of this matrix is
\bea P_{ij} + \lambda_{i}Q_{ij}  , \label{eq:MSR_recon_ij_ele} \eea
while the $(j,\ i)^{\text{th}}$ element is given by
\bea
&P_{j i} + \lambda_{j}Q_{j i} \nonumber \\
&= P_{i j}+\lambda{j}Q_{i j}~, \label{eq:MSR_recon_ji_ele}
\eea
where equation (\ref{eq:MSR_recon_ji_ele}) follows from the symmetry of $P$ and $Q$.  By construction, all the $\lambda_i$ are distinct and hence using  (\ref{eq:MSR_recon_ij_ele}) and (\ref{eq:MSR_recon_ji_ele}), the data collector can solve for the values of $P_{ij}, \ Q_{ij}$ for all $i \neq j$.

Consider first the matrix $P$. Let $\Phi_{\text{\tiny DC}}$ be given by
\beq \Phi_{\text{\tiny DC}} = \left[\begin{tabular}{>{$}c<{$}}
                               \underline{\phi}^t_1 \\
                \vdots \\
                \underline{\phi}^t_{\alpha+1}
                              \end{tabular}\right].
\eeq
All the non-diagonal elements of $P$ are known.  The elements in the $i^{th}$ row (excluding the diagonal element) are given by
\beq \underline{\phi}_i^t  S_{1} \left[ \underline{\phi}_1 \cdots \ \underline{\phi}_{i-1} \ \underline{\phi}_{i+1} \cdots \ \underline{\phi}_{\alpha+1} \right] .\eeq
However, the matrix to the right is non-singular by construction and hence the data collector can obtain
\beq \left\lbrace \underline{\phi}_i^t  S_{1} ~~|~~ 1\leq i \leq \alpha+1\right\rbrace .  \eeq
Selecting the first $\alpha$ of these, the data collector has access to
\beq  \left[\begin{tabular}{>{$}c<{$}}
\underline{\phi}^t_1 \\
\vdots \\
\underline{\phi}^t_{\alpha}
\end{tabular}\right] S_{1}~. \eeq
The matrix on the left is also non-singular by construction and hence the data collector can recover $S_{1}$. Similarly, using the values of the non-diagonal elements of $Q$, the data collector can recover $S_{2}$.
\end{IEEEproof}

\begin{remark}It is shown in~\cite{ourITW,ourMISERjournal} that \textit{interference alignment} is, in fact, a necessary ingredient of any minimum storage regenerating code. Interference alignment is also present in the product-matrix MSR code, and Appendix~\ref{app:ia} brings out this connection. 
\end{remark}

\subsection{An Example for the Product-Matrix MSR code}\label{sec:MSR_example}
Let $n=6, \ k=3, \ d=4$. Then $\alpha=d-k+1=2$ and $B =k \alpha =  6$. Let us choose $q=13$, so we are operating over $\mathbb{F}_{13}$. The matrices $S_1$ and $S_2$ are filled up by the $6$ message symbols $\{u_i\}_{i=1}^6$ as shown below:

\beq S_1= \begin{bmatrix}
      u_1 & u_2  \\
           u_2 & u_3
     \end{bmatrix}, S_2 = \begin{bmatrix}
      u_4 & u_5  \\
           u_5 & u_6
     \end{bmatrix}, \eeq
so that the message matrix $M$ is given by
\beq M= \begin{bmatrix}
         u_1 & u_2  \\
           u_2 & u_3 \\
    u_4 & u_5  \\
           u_5 & u_6
        \end{bmatrix}. \eeq
We choose $\Psi$ to be the $(6 \times 4)$ Vandermonde matrix over $\mathbb{F}_{13}$ given by \beq \Psi= \begin{bmatrix}
            1 & 1 & 1 & 1\\
            1 & 2 & 4 & 8\\
            1 & 3 & 9 & 1\\
            1 & 4 & 3 & 12\\
            1 & 5 & 12 & 8\\
            1 & 6 & 10 & 8
           \end{bmatrix}.  \eeq
Hence the $(6 \times 2)$ matrix $\Phi$ and the $(6 \times 6)$ diagonal matrix $\Lambda$ are
\beq \Phi=\begin{bmatrix}
1 & 1 \\
 1 & 2 \\
1 & 3 \\
1 & 4 \\
1 & 5 \\
1 & 6
\end{bmatrix}, \quad \Lambda=\begin{bmatrix}
1 & & & & & \\
& 4 & & & &\\
& & 9 & & &\\
& & & 3 & &\\
& & & & 12 &\\
& & & & & 10
\end{bmatrix}. \eeq
Fig.~\ref{fig:MSR_example} shows at the top, the $(6 \times 2)$ code matrix $C=\Psi M$ with entries expressed as functions of the message symbols $\{u_i \}$.  The rest of the figure explain how exact-regeneration of failed node $1$ takes place. To regenerate node node $1$, the helper nodes (nodes $2,4,5,6$ in the example), pass on their respective inner products $\underline{\psi}_{\ell}^t M [1 \ 1]^t$ for $\ell=2,4,5,6$.  The replacement node multiplies the symbols it receives with $\Psi_{\text{repair}}^{-1}$, where
           \beq
           \Psi_{\text{repair}} \ = \
           \begin{bmatrix}
            1 & 2 & 4 & 8\\
            1 & 4 & 3 & 12\\
            1 & 5 & 12 & 8\\
            1 & 6 & 10 & 8
           \end{bmatrix},  \eeq
and decodes $S_1 \underline{\psi}_1$ and $S_2 \underline{\psi}_1$:
\beq S_1 \underline{\psi}_1 = \begin{bmatrix}
u_1+u_2 \\
u_2+ u_3
\end{bmatrix}, S_2 \underline{\psi}_1 = \begin{bmatrix}
u_4+u_5 \\
u_5+ u_6
\end{bmatrix}. \eeq
Finally, it processes $S_1 \underline{\psi}_1$ and $S_2 \underline{\psi}_1$ to obtain the data stored in the failed node as explained in the proof of Theorem~\ref{thm:MSR_regen} above.

\begin{figure*}[t]
% \hspace{20pt}
\centering
\includegraphics[trim=0.4in 5.5in 4.5in 0.5in, clip, width=.85\textwidth]{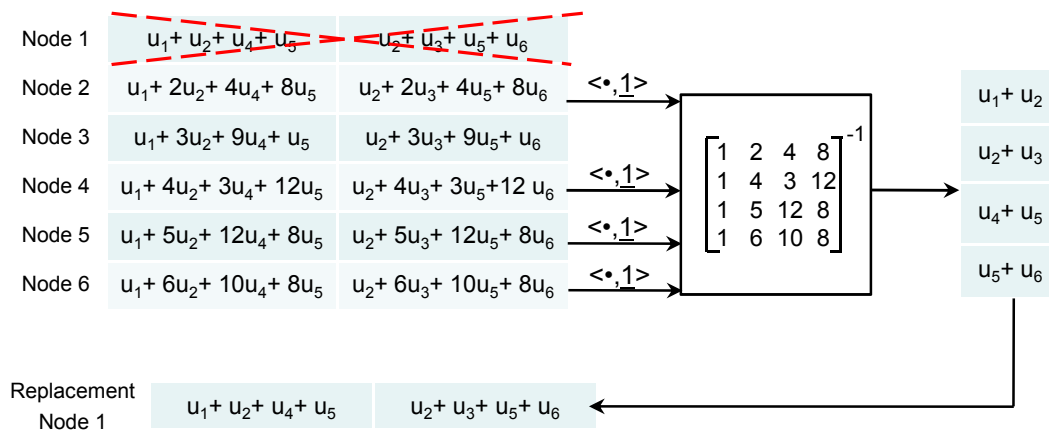} %l b r t
\caption{An example for the MSR code construction: On failure of node $1$, the replacement node downloads one symbol each from nodes $2, \ 4, \ 5$ and $6$, using which node $1$ is exactly regenerated. The notation $<.,\underline{1}>$ indicates an inner product of the stored symbols with the vector $[1~1]^t$.} \label{fig:MSR_example}
% \vspace{-10pt}
\end{figure*}

\subsection{Systematic Version of the Code}\label{subsec:MSR_sys}

It was pointed out in Section~\ref{sec:prod_mx}, that every exact-regenerating code has a systematic version and further, that the code could be made systematic through a process of message-symbol remapping. In the following, we make this more explicit in the context of the product-matrix MSR code.

Let $\Psi_k$ be the $(k \times d)$ submatrix of $\Psi$, containing the $k$ rows of $\Psi$ corresponding to the $k$ nodes which are chosen to be made systematic. The set of $k\alpha$ symbols stored in these $k$ nodes are given by the elements of the $(k \times \alpha)$ matrix $\Psi_k M$. Let $U$ be a $(k \times \alpha)$ matrix containing the $B=k\alpha$ source symbols.  We map \beq \Psi_k M = U~, \eeq and solve for the entries of $M$ in terms of the symbols in $U$.  This is precisely the data-reconstruction process that takes place when a data collector connects to the $k$ chosen nodes. Thus, the value of the entries in $M$ can be obtained by following the procedure outlined in Theorem~\ref{thm:MSR_recon}. Then, use this $M$ to obtain the code $C = \Psi M$. Clearly, in this representation, the $k$ chosen nodes store the source symbols $U$ in uncoded form.

\subsection{Explicit MSR Product-Matrix Codes for $d \geq 2k-2$} \label{sec:MSR_extension}

In this section, we show how an MSR code for $d=2k-2$ can be used to
obtain MSR codes for all $d\geq 2k-2$.  Our starting point is the
following theorem.

%TODO: say that the given code C is linear, in each thm below
\begin{thm}
\label{thm:MSR_higher_d} An explicit $[n'=n+1,~k'=k+1,~d'=d+1]$ exact-regenerating code ${\cal C}'$ that achieves the cut-set bound at the MSR point can be used to construct an explicit $[n,\ k, \ d]$ exact-regenerating code ${\cal C}$ that also achieves the cut-set bound at the MSR point. Furthermore if $d'=a k'+b$ in code ${\cal C}'$, $d=a k+b+(a-1)$ in code ${\cal C}$.   If ${\cal C}'$ is linear, so is ${\cal C}$.
\end{thm}

\begin{IEEEproof}
If both codes operate at the MSR point, then the number of message symbols $B',B$ in the two cases must satisfy \bean  B'
\ = \  k'(d'-k'+1) \ \ \text{  and } \ \ B \ = \ k(d-k+1) \eean respectively, so that
\[
B'-B \ = \ d-k+1 = \alpha~.
\]
We begin by constructing an MSR-point-optimal $[n', ~k',~ d']$ exact-regenerating code ${\cal C}'$ in systematic form with the first $k'$ rows containing the $B'$ message symbols.  Let ${\cal C}''$ be the subcode of ${\cal C}'$ consisting of all code matrices in ${\cal C}'$ whose top row is the all-zero row (i.e., the first $\alpha$ of the $B'$ message symbols are all zero).  Clearly, the subcode ${\cal C}''$ is of size $q^{B'-\alpha} \ = \ q^B $.  Note that ${\cal C}''$ also possesses the same exact-regeneration and data-reconstruction properties as does the parent code ${\cal C}'$.

%The code ${\cal C}'$ is comprised of $q^B$ code matrices each of size $(n+1 \times \alpha)$ and we can assume without loss of generality that the first rows of all code matrices are identical and that this common row $\underline{\rho}$ contains the $\alpha$ all-zero fictitious message symbols.

Let the code ${\cal C}$ now be formed from subcode ${\cal C}''$ by puncturing (i.e., deleting) the first row in each code matrix of ${\cal C}''$.
Clearly, code ${\cal C}$ is also of size $q^B$.  We claim that ${\cal C}$ is an $[n,k,d]$ exact-regenerating code. The data-reconstruction requirement requires that the $B$ underlying message symbols be recoverable from the contents of any $k$ rows of a code matrix $C$ in ${\cal C}$.  But this follows since, by augmenting the matrices of code ${\cal C}$ by placing at the top an additional all-zero row, we obtain a code matrix in ${\cal C}''$ and code ${\cal C}''$ has the property that the data can be recovered from any $(k+1)$ rows of each code matrix in ${\cal C}''$.  A similar argument shows that code ${\cal C}$ also possesses the exact-regeneration property.  Clearly if ${\cal C}'$ is linear, so is code ${\cal C}$.  Finally, we have
\bean
d' & = & a k'+b \\
\ \Rightarrow \ d+1 & = & a(k+1) + b \\
\ \Rightarrow \ d & = & a k + b + (a-1)~. \eean
\end{IEEEproof}

By iterating the procedure in the proof of Theorem~\ref{thm:MSR_higher_d} above $i$ times we obtain:

\begin{cor}
\label{cor:MSR_higher_d} An explicit $[n'=n+i,~k'=k+i,~d'=d+i]$
exact-regenerating code ${\cal C}'$ that achieves the cut-set bound at the MSR point can be used to construct an explicit
$[n,\ k, \ d]$ exact-regenerating code ${\cal C}$ that also achieves the cut-set bound at the MSR point.
Furthermore if $d'=a k'+b$ in code ${\cal C}'$, $d=a k+b+i(a-1)$ in code
${\cal C}$.   If ${\cal C}'$ is linear, so is ${\cal C}$.
\end{cor}

The corollary below follows from Corollary~\ref{cor:MSR_higher_d} above.
\begin{cor}
\label{cor:MSR_higher_d_2} An MSR-point optimal exact-regenerating code
${\cal C}$ with parameters $[n,k,d]$ for any $2k-2 \leq d \leq n-1$
can be constructed from an MSR-point optimal exact-regenerating
$[n'=n+i,k'=k+i,d'=d+i]$ code ${\cal C}'$ with $d'=2k'-2$ and
$i=d-2k+2$. If ${\cal C}'$ is linear, so is ${\cal C}$.
\end{cor}

\section{Analysis and Advantages of the Codes} \label{sec:sys_adv}
In this section, we detail the system-implementation advantages of
the two code constructions presented in the paper.

\subsection{Reduced Overhead}
In the product-matrix based constructions provided, the data stored in the $i^{th}$ storage node in the system is completely determined by the single encoding vector $\underline{\psi}_i$ of length $d$. This is in contrast to a $(B \times \alpha)$ generator matrix in a general code, comprising of the $\alpha$ global kernels of length $B$, each associated to a different symbol stored in the node. The encoding vector suffices for the encoding, data-reconstruction, and regeneration purposes. The short length of the encoding vector reduces the overhead associated with the need for nodes to communicate their encoding vectors to the data collector during data-reconstruction, and to the replacement node during regeneration of a failed node.

Also, in both MBR and MSR code constructions, during regeneration of a failed node, the information passed on to the replacement node by a helper node is only a function of the index of the failed node. Thus, it is independent of the identity of the $d-1$ other nodes that are participating in the regeneration. Once again, this reduces the communication overhead by requiring less information to be disseminated.

\subsection{Applicability to Arbitrary $n$}
In any real-world distributed storage application such as peer-to-peer storage, cloud storage, etc, it is natural that the number of nodes may go up or down: in due course of time, new nodes may be added to the system, or multiple nodes may fail or exit the system. For example, in peer-to-peer systems, individual nodes are free to come and go at will. The existing, explicit constructions of exact-regenerating codes \cite{ourAllerton,DimSearch,WuDimISIT,ourITW,Changho} restrict the value of $n$ to be $d+1$. On the other hand, the codes presented in this paper are applicable for all values of $n$, and independent of the values of the parameters $k$ and $d$. This gives a practical appeal to the code constructions presented here. 

\subsection{Complexity}\label{subsec:complexity}
\subsubsection{Linearity and Field Size} The codes are linear over a chosen finite field $\mathbb{F}_q$, i.e., the source symbols are from this finite field, and any stored symbol is a linear combination of these symbols over $\mathbb{F}_q$.
%Such a setup makes the code easy to implement.
To arrive at the product-matrix MBR code, any field of size $2n$ or higher suffices, and for the product-matrix MSR code, any field of size $n^2$ or higher suffices. By cleverly choosing the matrix $\Psi$ that meets the conditions governing the respective codes, it may often be possible to reduce the field size even further.

\subsubsection{Striping} The codes presented here divide the entire message into \textit{stripes} of sizes corresponding to $\beta=1$. Since each stripe is of minimal size, the complexity of encoding, data-reconstruction and regeneration operations, are considerably lowered, and so are the buffer sizes required at data collectors and replacement nodes. Furthermore, the operations that need to be performed on each stripe are identical and independent, and hence can be performed in parallel efficiently by a GPU/FPGA/multi-core processor.

\subsubsection{Choice of the Encoding Matrix $\Psi$}
The encoding matrix $\Psi$, for both the codes described, can be chosen as a Vandermonde matrix. Then each encoding vector can be described by just a scalar. Moreover with this choice, the encoding, data-reconstruction, and regeneration operations are, for the most part, identical to encoding or decoding of conventional Reed-Solomon
codes.

\section{Conclusions}\label{sec:conclusion}
In this paper, an explicit MBR code for all values of the system parameters $[n, \ k, \ d]$, and an explicit MSR code for all parameters satisfying $[n,~k,~d \geq 2k-2]$ are presented. Both constructions are based on a common product-matrix framework introduced in this paper, and possess attributes that make them attractive from an implementation standpoint. To the best of our knowledge, these are the first explicit constructions of exact-regenerating codes that allow $n$ to take any value independent of the other parameters; this results in a host of desirable properties such as the ability to optimally handle multiple simultaneous, node failures as well as the ability of  allowing the total number of storage nodes in the system to vary with time. Our results also prove that the MBR point on the storage-repair bandwidth tradeoff is achievable under the additional constraint of exact-regeneration for all values of the system parameters, and that the MSR point is achievable under exact-regeneration for all $d \geq 2k-2$.

\bibliographystyle{IEEEtran}
\bibliography{../bibtex/distributedStorage,distributedStorage}

\appendices

\section{Description of a Previously Constructed MSR Code in the Product-Matrix Framework}\label{app:MISER_prodmx}

A code structure that guarantees exact-regeneration of just the systematic nodes is provided in~\cite{ourITW}, for the MSR point with parameters $[n=d+1,~k,~d\geq 2k-1]$.  Subsequently, it was shown in~\cite{Changho} that for this set of parameters, the code introduced in~\cite{ourITW} for exact-regeneration of only the systematic nodes can also be used for exact-regeneration of the non-systematic (parity) nodes, provided repair construction schemes are appropriately designed. Such an explicit repair scheme is indeed designed and presented in \cite{Changho}. In this section, we provide a simpler description of this code in the product-matrix framework.

As in~\cite{ourITW,Changho}, we begin with the case $d=2k-1$, since the code as well as both data-reconstruction and exact-regeneration algorithms can be extended to larger values of $d$ by making use of Corollary~\ref{cor:MSR_higher_d_2}.

At the MSR point, with $d=n-1=2k-1$, we have from equations (\ref{eq:MSR_alpha_beta1}) and (\ref{eq:MSR_B_beta1}) that
\bea \alpha &=& d-k+1 \ = \ k~, \\
 B &=& k \alpha = k^2~.\eea
Let $S$ be a $(k \times k)$ matrix whose entries are precisely the $B$ message symbols $\{u_i\}_{i=1}^B$ and let $M$ be the $(2k \times k)$ message matrix\footnote{Note that the constructions presented in Sections~\ref{sec:MBR_prodmx} and~\ref{sec:MSR_prodmx} employ a $(d \times \alpha)$ matrix $M$ as the message matrix, whereas the dimension of $M$ in the present construction is $((d+1) \times \alpha)$.}
given by:  \beq M \ = \
\begin{bmatrix}
S  \\
S^t
\end{bmatrix}. \eeq
Next, let $\Phi$ be a $(k \times k)$ Cauchy matrix over $\mathbb{F}_q$ and $\rho$ a scalar chosen such that \bea \rho \neq 0, \quad \rho^2 \neq 1~. \eea   Let $\Psi$ be the $(n \times 2k)$ encoding matrix given by \bea \Psi =
\begin{bmatrix}
I & 0 \\
\Phi & \rho \Phi \\
\end{bmatrix} . \eea

The code constructed in~\cite{ourITW,Changho} can be verified to have an alternate description  as the collection of code matrices of the form \beq
C \ = \Psi M =
\begin{bmatrix}
S \\
\Phi(S +\rho S^t)
\end{bmatrix}. \eeq
Note that the first $k$ nodes store the message symbols in uncoded form and hence correspond to the systematic nodes.  A simple description of the exact-regeneration and data-reconstruction properties of the code is presented below.

\begin{thm}[Exact-Regeneration] In the code presented, exact-regeneration of any failed node can be achieved by connecting to the remaining $n-1$ nodes.
\end{thm} \label{thm:MISER_regen}
\begin{IEEEproof}
In this construction, the vector $\underline{\mu}_f$ used in the exact-regeneration of a failed node $f$ is composed of the first $k=\alpha$ symbols of $\underline{\psi}_f$.
\paragraph{Exact-regeneration of systematic nodes}
Consider regeneration of the $i^{th}$ systematic node. The $k$ symbols thus desired by the replacement node are $\underline{e}_i^t S$. The replacement node obtains the following $n-1$ symbols from the remaining nodes: \beq \begin{bmatrix}
\tilde{I} & 0 \\
\Phi & \rho \Phi \\
\end{bmatrix} \begin{bmatrix}
S \\
S^t
\end{bmatrix} \underline{e}_i = \begin{bmatrix}
\tilde{I}S\underline{e}_i \\
\Phi ({S}+\rho {S^t})\underline{e}_i
\end{bmatrix}~. \label{eq:miser1}\eeq
where $\tilde{I}$ is a $((k-1) \times k)$ matrix which is the identity matrix with $i^{th}$ row removed. Since $\Phi$ is full rank by construction, the replacement node has access to \bea &&\left[({S}+\rho {S^t})\underline{e}_i\right] \nonumber\\ &&=\left[\rho\underline{e}_i^t{S}+\underline{e}_i^t {S^t}\right]~.\label{eq:miser2} \eea 
From~\eqref{eq:miser1} and~\eqref{eq:miser2}, we see that the replacement node has access to
\beq \begin{bmatrix}
\tilde{I} & 0 \\
\Phi & \rho \Phi \\
\rho\underline{e}_i^t & \underline{e}_i^t
\end{bmatrix} \begin{bmatrix}
S \\
S^t
\end{bmatrix} \underline{e}_i~. \eeq
Since $\rho \neq 1$, the $(2k \times 2k)$ matrix on the left is non-singular. This allows the replacement node to recover the symbols $S^t \underline{e}_i$, which are precisely the set of symbols $\underline{e}_i^t S$ desired.

\paragraph{Exact-regeneration of non-systematic nodes}
Let $\underline{\phi}_f^t$ be the row of $\Phi$ corresponding to the failed node. Then the $k$ symbols stored in the failed node are $\underline{\phi}_f^t(S +\rho S^t)$. The replacement node requests and obtains the following $n-1$ symbols from the remaining nodes:
\beq \begin{bmatrix}
      I & 0 \\
\Phi_{k-1} & \rho \Phi_{k-1} \\
     \end{bmatrix} \begin{bmatrix}
S \\
S^t
\end{bmatrix} \underline{\phi}_f~, \eeq
where $\Phi_{k-1}$ is the submatrix of $\Phi$ containing the $k-1$ rows corresponding to the remaining non-systematic nodes. This gives  the replacement node access to  $S \underline{\phi}_f$ and therefore to \beq (S \underline{\phi}_f)^t \underline{\phi}_f = \underline{\phi}_f^t
S^t \underline{\phi}_f~. \eeq Hence the replacement node has access to
\beq \begin{bmatrix}
I & 0 \\
\Phi_{k-1} & \rho \Phi_{k-1} \\
\underline{0}^{t} & \underline{\phi}_f^t
\end{bmatrix} \begin{bmatrix}
S \\
S^t
\end{bmatrix} \underline{\phi}_f~. \eeq
The matrix on the left is easily verified to be non-singular and thus the replacement node acquires $S \underline{\phi}_f$ and $S^t \underline{\phi}_f$ individually from which it can derive the desired vector $(\underline{\phi}_f^t S + \rho \underline{\phi}_f^t S^t)$.
\end{IEEEproof}

\begin{thm}[Data-Reconstruction]\label{thm:MISER_recon}
In the code presented, all the $B$ message symbols can be recovered by connecting to any $k$ nodes, i.e., the message symbols can be recovered through linear operations on the entries of any $k$ rows of the matrix $C$.
\end{thm}
\begin{IEEEproof}
We first introduce the following notation to denote submatrices of a matrix. If $A$ is an $(m_1 \times m_2)$ matrix and $P,~Q$ are arbitrary subsets of $\{1,\ldots,m_1\}$ and $\{1,\ldots,m_2\}$ respectively, we will use $A_{(P,Q)}$ to denote the submatrix of $A$ containing only the rows and columns respectively specified by the indices in $P$ and $Q$. For the cases when either $P=\{1,\ldots,m_1\}$ or $Q=\{1,\ldots,m_2\}$, we will simply indicate this as `all'.

Let $P=\{n_1, \ldots, n_i\}$ and $Q=\{m_1, \ldots, m_{(k-i)}\}$ be the systematic and non-systematic nodes respectively to which the data collector connects.  Let $T=\{1,\ldots,k\} \backslash P$, i.e., the systematic nodes to which the data collector does \textit{not} connect. Then the data collector is able to access the $k \alpha$ symbols
\beq \begin{bmatrix}
S_{(P,\text{all})} \\
\Phi_{(Q,\text{all})}(S +\rho S^t) \\
\end{bmatrix}. \eeq
Thus the data collector has access to the $i$ rows of $S$ indexed by the entries of $P$ and consequently, has access to the corresponding columns of $S^t$ as well.

Consider the $i$ columns of $\Phi_{(Q,\text{all})}(S +\rho S^t)$ indexed by $P$. Since the entries of these columns in $S^t$ are known, the data collector has access to $\Phi_{(Q,\text{all})} S_{(\text{all},P)}~.$ Now since the $i$ rows of $S$ indexed through $P$ are also known, the data collector has thus access to the product \bea \Phi_{(Q,T)} S_{(T,P)}~.\eea
% \text{\tiny $(k-i \times k-i)$} \ \ \text{\tiny $(k-i \times k-i)$} \nonumber \eea
Now as $\Phi_{(Q,T)}$ is non-singular, being a $(k-|P|, \ k- |P|)$ sub-matrix of a Cauchy matrix, the data collector can recover $S_{(T,P)}$. In this way, the data collector has recovered all the entries in the rows of $S$ indexed by $P$ as well as all the entries in the columns of $S$ indexed by $P$.  Clearly, the same statement holds when $S$ is replaced by $S^t$. Thus the data collector has access to the product:
\beq \Phi_{(T,T)}(S +\rho S^t)_{(T,T)}~.\eeq
Again $\Phi_{(T,T)}$ being a sub-matrix of a Cauchy matrix is of full rank and enables the data collector to recover $(S +\rho S^t)_{(T,T)}$. It is easy to see that from the diagonal elements of this matrix, all the diagonal elements of $S_{(T,T)}$ can be obtained. The non-diagonal elements are however of the form $S_{l j}+\rho S_{j l}$ and $S_{j l}+\rho S_{l j}$ for $l \in T, \ j \in T, \ l \neq j$. As $\rho^2 \neq 1$, all the non-diagonal elements of $S_{(T,T)}$ can also be decoded. In this way, the data collector has recovered all the $B$ entries of $S$. \end{IEEEproof}

\section{Equivalent Codes and Conversion of Non-Systematic Codes to Systematic}
\label{app:gen_mx}
In this section, we define the notion of `equivalent codes', and show that any exact-regenerating code is equivalent to a systematic exact-regenerating code.

Given any linear exact-regenerating code, one can express each of the $n\alpha$ symbols stored in the nodes as a linear combination of the $B$ message symbols $\{u_i\}_{i=1}^{B}$. Let $\{c_{ij} | 1 \leq i \leq n,~1 \leq j \leq \alpha\}$ denote the $j^{th}$ symbol stored in the $i^{th}$ node. Thus we have the relation: 
\bea [u_1 ~ u_2 ~ \cdots ~u_B] [ G_1 ~ G_2 ~ \cdots~ G_n ] \ = \ \left[ c_{11} \cdots
c_{1\alpha} | c_{21} \cdots c_{2\alpha} | \cdots \cdots |
c_{n1} \cdots c_{n\alpha} \right]~,
 \eea
where the $(B \times n\alpha)$ block generator matrix $G=[ G_1 \ G_2 \ \cdots G_n]$ is composed of the $n$ component generator sub-matrices
\[
G_i \ = \  \left[ \begin{array}{cccc} \underline{g}_{i1} &
\underline{g}_{i2} & \cdots & \underline{g}_{i \alpha}
\end{array} \right],
\]
each of size $(B \times \alpha)$, and associated to a distinct node.\footnote{In the terminology of network coding, the $(B \times 1)$ column vector $\underline{g}_{ij}$ is termed the $j^{th}$ global kernel associated to the $i^{th}$ node.} 
Let $W_i$ denote the column-space of $G_i$. A little thought will show that a distributed storage code is an exact-regenerating code iff 
\ben
\item for every subset of $k$ nodes $\{i_j \mid 1 \leq j \leq k\}$,
\[
\text{dim}(W_{i_1}+W_{i_2}+\cdots+W_{i_k}) \ = \ B~,\]
 and
\item for every subset of $(d+1)$ nodes $\{i_j \mid 1 \leq j \leq (d+1)\}$, the subspaces
$\{W_{i_j}\}_{j=1}^d$ contain a vector $\underline{w}_j$ such that
\[
W_{i_{d+1}} \ \subseteq \ \text{span}\left( \underline{w}_{i_1},
\underline{w}_{i_2}, \cdots, \underline{w}_{i_d} \right) .
\]
\een
We can thus define two exact-regenerating codes to be \textit{equivalent} if the associated subspaces $\{W_i\}_{i=1}^n$ are identical. It is also clear that two codes are equivalent if one can be obtained from the other through a non-singular transformation of the message symbols and the symbols stored within the nodes. With these two observations, it follows that two codes with generator matrices having the following relation are equivalent:
\[
G, \ \ \ \text{and} \ \  \ X G\left[ \begin{array}{cccc} Y_1 & & & \\
& Y_2 & & \\
&  & \ddots & \\
&  & & Y_n \end{array} \right],
\]
where the $(B \times B)$ pre-multiplication matrix $X$, and the $(n\alpha \times n\alpha)$ post-multiplication block diagonal matrix comprising of the $(\alpha \times \alpha)$ matrices $\{Y_i\}_{i=1}^{n}$, are non-singular. Clearly, equivalent codes have identical data-reconstruction and regeneration properties.

\paragraph*{Systematic Version of Exact-Regenerating Codes}
It also follows that any exact-regenerating code is equivalent to a systematic, exact-regenerating code. To see this, suppose the set of $k$ nodes to be systematic are the first $k$ nodes. Let \[ \{
\underline{\tilde{g}}_{a_1}, \underline{\tilde{g}}_{a_2}, \cdots, \underline{\tilde{g}}_{a_B}
\} \ \subseteq \ \{ \underline{g}_{ij} \mid 1 \leq i \leq k, \ 1
\leq j \leq \alpha \}
\]
denote a set of $B$ linearly independent column vectors drawn from the generator matrices of the first $k$ nodes $[G_1~\cdots~G_k]$. That such a subset is guaranteed to exist follows from the data-reconstruction property of a regenerating code. Let $\tilde{G}$ be the $(B \times B)$ invertible matrix \[ \tilde{G} \ = \ \left[ \begin{array}{cccc} \underline{\tilde{g}}_{a_1} &  \underline{\tilde{g}}_{a_2} & \cdots & \underline{\tilde{g}}_{a_B} \end{array} \right] . \]
Then we have the relation: \bea [u_1~ u_2~ \cdots~ u_B]\tilde{G} & = &
[\tilde{c}_{a_1} ~\tilde{c}_{a_2}~  \cdots ~ \tilde{c}_{a_B}]~, \eea where $\{\tilde{c}_{a_i}\}_{i=1}^{B}$ is the corresponding set of code symbols. It follows that if we
wish to encode in such a way that the code is systematic with
respect to code symbols $\{\tilde{c}_{a_i}\}_{i=1}^{B}$, the
input to be ``fed'' to the generator matrix $G$ is
\[
[u_1~ u_2~ \cdots~ u_B] \tilde{G}^{-1}~.
\]
%It follows that the generator matrix
%\[
%G' \ = \ A^{-1} G
%\]
%is the generator matrix of an equivalent code in which the code
%symbols $[ c_{a1}, \ c_{a2}, \ \cdots, \ c_{aB} ]$ are precisely the
%message symbols themselves.

\section{Interference Alignment in the Product-Matrix MSR Code}\label{app:ia}
The concept of \textit{interference alignment} was introduced in~\cite{MotKhan_IA,CadJafar_IA} in the context of wireless communication. This concept was subsequently used to construct regenerating codes in~\cite{WuDimISIT,ourITW,Changho,ourMISERjournal}. Furthermore,~\cite{ourITW,ourMISERjournal} showed that interference alignment is in fact, a necessary ingredient of any linear MSR code. Since the product-matrix MSR construction provided in the present paper does not explicitly use the concept of interference alignment, a natural question that arises is how does interference alignment manifest itself in this code. We answer this question in the present section.

Consider repair of a failed node (say, node $f$) in a distributed storage system employing an MSR code, and let nodes $\{1,\ldots,d\}$ be the set of $d$ helper nodes. Recall that (from equation~\eqref{eq:MSR_parameters}), at the MSR point we have $B=k\alpha$. Further, since all the $B$ message symbols should be recoverable from any subset of $k$ nodes, it must be that any subset of $k$ nodes does not store any redundant information. Let $\underline{c}_i$, $1 \leq i \leq n$, be an $\alpha$-length vector denoting the $\alpha$ symbols stored in node~$i$. Then, from the above argument, it is clear that any symbol in the system can be written as a linear combination of the $B$ symbols in $\{\underline{c}_f, \; \underline{c}_1, \; \ldots,\; \underline{c}_{(k-1)}\}$.

Let $\theta_{\ell}$, $k \, \leq \, \ell \, \leq \, d$, denote the symbol passed by node~$\ell$ to assist in the repair of node $f$. Then we can write,
\beq \theta_{\ell}= \underline{c}_f^t \, \underline{v}_{\ell,f} + \sum^{k-1}_{i=1}   \underline{c}^t_i \, \underline{v}_{\ell, i}\eeq
for some vectors $\underline{v}_{\ell, i}$ and $\underline{v}_{\ell,f}$ each of length $\alpha$. The symbols in $\{\underline{c}_f, \; \underline{c}_1, \; \ldots,\; \underline{c}_{(k-1)}\}$ have no redundancy among themselves. Thus, the components comprising of $\{\underline{c}_1, \; \ldots,\; \underline{c}_{(k-1)}\}$ are undesired and hence are termed as \textit{interference components}, and the component comprising of $\underline{c}_f$ is termed the \textit{desired component}.

It is shown in~\cite{ourITW,ourMISERjournal} that for any MSR code, it must be that for every $i \in \{1,\, \ldots, \,k-1\}$, the set of vectors
\beq \left\lbrace \underline{v}_{\ell,i} | k \leq \ell \leq d \right\rbrace \eeq
are \textit{aligned} (i.e., are scalar multiples of each other).

The following lemma considers the repair scenario discussed above to illustrate how interference alignment arises in the product-matrix MSR code presented in Section~\ref{sec:MSR_prodmx}.
\begin{lem}
For every helper node $\ell$, \, $k \leq \ell \leq d$, there exist scalars $\{a_{(\ell,i)} | 1 \leq i \leq k-1\}$ and an $\alpha$-length vector $\underline{b}_\ell = [b_{\ell,1} \cdots b_{\ell,\alpha}]^t$ such that
\beq \underline{\psi}_\ell^t M \underline{\phi}_f = \underline{\psi}_f^t M \underline{b}_\ell + \sum_{i=1}^{k-1} a_{\ell,i} \underline{\psi}_i^t M \underline{\phi}_f~.\eeq
\end{lem}
\begin{proof}
Re-writing the symbols passed by the helper node $j$ ($1 \leq j \leq d$),
\bea \underline{\psi}_j^t M \underline{\phi}_f 
&=& \left[\underline{\phi}_j^t S_1 + \lambda_j \underline{\phi}_j^t S_2\right] \underline{\phi}_f\\
&=& \left[\underline{\phi}_f^t S_1 + \lambda_j \underline{\phi}_f^t S_2\right] \underline{\phi}_j \label{eq:interference_alignment1}\\
&=& \left[\underline{\phi}_f^t S_1 + \lambda_f \underline{\phi}_f^t S_2\right] \underline{\phi}_j + (\lambda_j - \lambda_f) \underline{\phi}_f^t S_2 \underline{\phi}_j\\
&=& \underline{\psi}_f^t M \underline{\phi}_j + (\lambda_j - \lambda_f) \underline{\phi}_f^t S_2 \underline{\phi}_j~,\label{eq:interference_alignment2}
\eea
where equation~\eqref{eq:interference_alignment1} follows from the symmetry of matrices $S_1$ and $S_2$. By construction, the values of the scalars $\{\lambda_j \, | \, 1 \leq j \leq n\}$ are distinct, which allows us to write
\beq
\underline{\phi}_f^t S_2 \underline{\phi}_j = (\lambda_j - \lambda_f)^{-1} (\underline{\psi}_j^t M \underline{\phi}_f  - \underline{\psi}_f^t M \underline{\phi}_j)~.\label{eq:interference_alignment3}
\eeq
Also, since the $(k-1 =) \alpha$-length vectors $\{\underline{\phi}_i \, | \, 1\leq i \leq k-1\}$ are linearly independent, for $k \leq \ell \leq d$, there exist scalars $\{\tilde{a}_{\ell,i} \,| \, 1\leq i \leq k-1\}$ such that
\beq \underline{\phi}_\ell = \sum_{i=1}^{k-1} \tilde{a}_{\ell,i} \underline{\phi}_i~.\label{eq:interference_alignment4}\eeq
From equations~\eqref{eq:interference_alignment2},~\eqref{eq:interference_alignment3} and~\eqref{eq:interference_alignment4}, for any $\ell \in \{k, \; \ldots, \; d\}$, we can write
\bea \underline{\psi}_\ell^t M \underline{\phi}_f 
&\!\!=\!\!& \underline{\psi}_f^t M \underline{\phi}_\ell + (\lambda_\ell - \lambda_f) \underline{\phi}_f^t S_2 \underline{\phi}_\ell \label{eq:interference_alignment5} \\
 &=& \underline{\psi}_f^t M \underline{\phi}_\ell + (\lambda_\ell - \lambda_f) \sum_{i=1}^{k-1} \tilde{a}_{\ell,i} \underline{\phi}_f^t S_2 \underline{\phi}_i \label{eq:interference_alignment6} \\
% &=& \underline{\psi}_f^t M \underline{\phi}_\ell + (\lambda_\ell - \lambda_f) \sum_{i=1}^{k-1} \tilde{a}_{\ell,i} \underline{\phi}_f^t S_2 \underline{\phi}_i\\
% &=& \underline{\psi}_f^t M \underline{\phi}_\ell + (\lambda_\ell - \lambda_f) \sum_{i=1}^{k-1} \tilde{a}_{\ell,i} (\lambda_i - \lambda_f)^{-1} (\underline{\psi}_i^t M \underline{\phi}_f  - \underline{\psi}_f^t M \underline{\phi}_i)\\
&\!\!=\!\!& \underline{\psi}_f^t M \left( \underline{\phi}_\ell - (\lambda_\ell - \lambda_f) \sum_{i=1}^{k-1} \tilde{a}_{\ell,i} (\lambda_i - \lambda_f)^{-1} \underline{\phi}_i \right)  + \sum_{i=1}^{k-1}  \left( \tilde{a}_{\ell,i} (\lambda_\ell - \lambda_f)(\lambda_i - \lambda_f)^{-1}\right) (\underline{\psi}_i^t M \underline{\phi}_f) \label{eq:interference_alignment7}
\eea
where equation~\eqref{eq:interference_alignment6} follows from~\eqref{eq:interference_alignment4}, and equation~\eqref{eq:interference_alignment7} follows from~\eqref{eq:interference_alignment3}.
\end{proof}
%Thus, the components of the symbols passed by helper node $\ell$ along the symbols stored in node $j$ is $a_{\ell,i} \underline{\psi}_i^t M \underline{\phi}_f$ are scalar multiples of each other, and are hence \textit{aligned}.

\end{document}